\documentclass[prl,aps,amsmath,amssymb,reprint]{revtex4-1}

\usepackage{color}
\usepackage{amsmath}
\usepackage{amsthm}
\usepackage{bm}
\usepackage{multirow}
\usepackage{graphicx}
\newtheorem{thm}{Theorem}
\theoremstyle{definition}
\newtheorem{defn}[thm]{Definition}
\theoremstyle{plain}
\newtheorem{prop}[thm]{Proposition}
\theoremstyle{plain}

\theoremstyle{plain}

\theoremstyle{remark}
\newtheorem{rem}[thm]{Remark}

\draft 

\begin{document}

\title{Birth and stabilization of phase clusters by multiplexing of adaptive networks} 



\author{Rico Berner$^{1,2}$}
\email[]{rico.berner@physik.tu-berlin.de}
\author{Jakub Sawicki$^{1}$}
\author{Eckehard Sch\"oll$^{1}$}
\email[]{schoell@physik.tu-berlin.de}
\affiliation{$^{1}$Institut f\"ur Theoretische Physik, Technische Universit\"at Berlin, Hardenbergstr.\,36, 10623 Berlin, Germany}
\affiliation{$^{2}$Institut f\"ur Mathematik, Technische Universit\"at Berlin, Hardenbergstr.\,36, 10623 Berlin, Germany}

\date{\today}

\begin{abstract}
We propose a concept to generate and stabilize diverse partial synchronization patterns (phase clusters) in adaptive networks which are widespread in neuro- and social sciences, as well as biology, engineering, and other disciplines. We show by theoretical analysis and computer simulations that multiplexing in a multi-layer network with symmetry can induce various stable phase cluster states in a situation where they are not stable or do not even exist in the single layer. Further, we develop a method for the analysis of Laplacian matrices of multiplex networks which allows for insight into the spectral structure of these networks enabling a reduction to the stability problem of single layers. We employ the multiplex decomposition to provide analytic results for the stability of the multilayer patterns. As local dynamics we use the paradigmatic Kuramoto phase oscillator, which is a simple generic model and has been successfully applied in the modeling of synchronization phenomena in a wide range of natural and technological systems.

\end{abstract}

\maketitle 

Complex networks are an ubiquitous paradigm in nature and technology, with a wide field of applications ranging from physics, chemistry, biology, neuroscience, to engineering and socio-economic systems. Of particular interest are adaptive networks, where the connectivity changes in time, for instance, the synaptic connections between neurons are adapted depending on the relative timing of neuronal spiking~\cite{MAR97a,ABB00,CAP08a,Meisel2009,LUE16}. 
Similarly, chemical systems have been reported~\cite{JAI01}, where the reaction rates adapt dynamically depending on the variables of the system. Activity-dependent plasticity is also common in epidemics~\cite{GRO06b} and in biological or social systems~\cite{GRO08a}. Synchronization is an important feature of the dynamics in networks of coupled nonlinear oscillators \cite{PIK01,STR01a,ALB02a,NEW03,BOC18}. Various synchronization patterns are known, like cluster synchronization where the network splits into groups of synchronous elements \cite{DAH12}, or partial synchronization patterns like chimera states where the system splits into coexisting domains of coherent (synchronized) and incoherent (desynchronized) states~\cite{KUR02a,ABR04,PAN15}. These patterns were also explored in adaptive networks \cite{AOK09,AOK11,NEK16,KAS16a,BER19,BER19a,KAS17,KAS18,GUS15a,PIC11a,TIM14,REN07,AVA18,Papadopoulos2017,Kasatkin2018a,Kasatkin2019}. 
Furthermore, adapting the network topology has also successfully been used to control cluster synchronization in delay-coupled networks~\cite{LEH14}. 

Another focus of recent research in network science are multilayer networks, which are systems interconnected through different types of links \cite{BOC14,DE13,DE15,KIV14}. A prominent example are social networks which can be described as groups of people with different patterns of contacts or interactions between them~\cite{GIR02,Amato2017,Amato2017a}. Other applications are communication, supply, and transportation networks, for instance power grids, subway networks, or airtraffic networks~\cite{CAR13d}. In neuroscience, multilayer networks represent for instance neurons in different areas of the brain, neurons connected either by a chemical link or by an electrical synapsis, or the modular connectivity structure of brain regions~\cite{Meunier2010,BEN16,BAT17,Vaiana2018,RAM19,Ashourvan2019,Zhou2006a,Zhou2007,Wang2019}. A special case of multilayer networks are multiplex topologies, where each layer contains the same set of nodes, and only pairwise connections between corresponding nodes from neighbouring layers exist~\cite{ZHA15a,MAK16,JAL16,GHO16,LEY17a,AND17,GHO18,SEM18,MIK18,SAW18c,OME19,RYB19,NIK19,Blaha2019,Sevilla-Escoboza2016,Requejo2016,Pitsik2018,Leyva2018,Jalan2019,Frolov2018}.

In spite of the lively interest in the topic of adaptive networks, little is known about the interplay of adaptively coupled groups of networks~\cite{KAS18,MAS18,Goremyko2017}. Such adaptive multilayer or multiplex networks appear naturally in neuronal networks, e.g., in interacting neuron populations with plastic synapses but different plasticity rules within each population~\cite{Citri2008,Edelmann2017}, or affected by different mechanisms of plasticity~\cite{Zenke2015}, or the transport of metabolic resources~\cite{Virkar2016}. Beyond brain networks, coexisting forms of (meta)plasticity are investigated in neuro-inspired devices to develop artificially intelligent learning circuitry~\cite{John2018}.

In this Letter we show that a plethora of novel patterns can be generated by multiplexing adaptive networks. In particular, partial synchronization patterns like phase clusters and more complex cluster states which are unstable in the corresponding monoplex network can be stabilized, or even states which do not exist in the single-layer case for the parameters chosen, can be born by multiplexing. Thus our aim is to provide fundamental insight into the combined action of adaptivity and multiplex topologies. Hereby we elucidate the delicate balance of adaptation and multiplexing which is a feature of many real-world networks even beyond neuroscience~\cite{Gross2006,Shai2013,Wardil2014,Klimek2016}. As local dynamics we use the paradigmatic Kuramoto phase oscillator model, which is a simple generic model and has been successfully applied in the modeling of synchronization phenomena in a wide range of natural and technological systems \cite{BOC18}. 

A general multiplex network with $L$ layers each consisting of $N$ identical adaptively coupled phase oscillators is described by
\begin{align}
\label{eq:PhiDGL_general}
	\dot{\phi}_{i}^{\mu} &=\omega-\frac{1}{N}\sum_{j=1}^{N}\kappa_{ij}^\mu\sin(\phi_{i}^\mu-\phi_{j}^\mu+\alpha^{\mu\mu}) \nonumber \\[-.15cm]
	& \phantom{\quad}-\sum_{\nu=1,\nu\ne\mu}^L\sigma^{\mu\nu}\sin(\phi_{i}^\mu-\phi_{i}^\nu+\alpha^{\mu\nu}), \\
	\dot{\kappa}_{ij}^\mu&=-\epsilon\left(\kappa_{ij}^\mu+\sin(\phi_{i}^\mu-\phi_{j}^\mu+\beta^{\mu})\right), \nonumber
\end{align}
where $\phi_i^{\mu}\in [0,2\pi)$ represents the phase of the $i$\textsuperscript{th} oscillator ($i=1,\dots,N$) in the $\mu$\textsuperscript{th} layer ($\mu=1,\dots,L$), and $\omega$ is the natural frequency. The interaction between the phase oscillators within each layer is described by the coupling matrix elements $\kappa_{ij}^{\mu}\in[-1,1]$. The intra-layer coupling weights $\kappa_{ij}^\mu$ are determined adaptively, whereas the inter-layer coupling weights $\sigma^{\mu\nu} \ge 0$ are fixed. The parameters $\alpha^{\mu\nu}$ are the phase lags of the interaction \cite{SAK86}. The adaptation rate $0<\epsilon \ll1$ is assumed to be a small parameter separating the time scales of the slow dynamics of the coupling weights and the fast dynamics of the oscillatory
system. Using the neuroscience terminology, such an adaptation may be called plasticity~\cite{AOK11}. The phase lag parameter $\beta^{\mu}$ of the adaptation function $\sin(\phi_{i}^\mu-\phi_{j}^\mu+\beta^\mu)$ describes different rules that may occur in neuronal networks.
For instance, for $\beta^\mu=-\pi/2$, a Hebbian-like rule~\cite{HOP96,SEL02,AOK15} is obtained where the coupling $\kappa_{ij}$ is increasing between any two systems with close-by phases, i.e., $\phi_{i}-\phi_{j}\approx 0$ ~\cite{HEB49}.
If $\beta=0$, the link $\kappa_{ij}$ will be strengthened if the
$i$\textsuperscript{th} oscillator is advancing the $j$\textsuperscript{th}. Such a relationship is typical for spike-timing dependent plasticity in neuroscience \cite{CAP08a,MAI07,LUE16,POP13}. If $\beta=\pi/2$, an anti-Hebbian-like rule emerges.

Eq.\,\eqref{eq:PhiDGL_general} has been widely used as a paradigmatic model for adaptive networks \cite{AOK09,AOK11,NEK16,KAS16a,BER19,BER19a,KAS17,KAS18,GUS15a,PIC11a,TIM14,REN07,AVA18}. It generalizes the Kuramoto-Sakaguchi model with fixed coupling topology~\cite{KUR84,STR00,ACE05a,PIK08,OME12b}.

Let us note important properties of the model. First, $\omega$ can be set to zero without loss of generality due to the shift-symmetry of Eq.\,\eqref{eq:PhiDGL_general}, \textit{i.e.}, considering the co-rotating frame $\phi\to\phi+\omega t$. Moreover, due to the existence of the attracting region $G\equiv\left\{ \left(\phi_{i}^\mu,\kappa_{ij}^\mu\right):\phi_{i}^\mu\in (0,2\pi],|\kappa_{ij}^\mu|\le1,\,i,j=1,\dots,N,\right.$ $\left.\mu=1,\dots,L\right\} $, one can restrict the range of the coupling weights to the interval $-1\le\kappa_{ij}\le1$ \cite{KAS17}. Finally, based on the parameter symmetries of the model
\begin{align*}
	(\bm{\alpha},\bm{\beta},\bm{\phi},\bm{\kappa}) &\mapsto (-\bm{\alpha},\pi-\bm{\beta},-\bm{\phi},\bm{\kappa}),\\
	(\alpha^{\mu\mu},\beta^{\mu},\phi_{i}^{\mu},\kappa_{ij}^{\mu}) &\mapsto(\alpha^{\mu\mu}+\pi,\beta^{\mu}+\pi,\phi_{i}^{\mu},-\kappa_{ij}^{\mu}),
\end{align*}
where $\bm{\alpha},\bm{\beta},\bm{\phi},\bm{\kappa}$ abbreviate the whole set of variables and parameters, it is sufficient to analyze the system within the parameter region $\alpha^{11}\in[0,\pi/2)$, $\alpha^{\mu\mu}\in[0,\pi)$ ($\mu\ne 1$), $\alpha^{\mu\nu}\in[0,2\pi)$ ($\mu\ne\nu$) and $\beta^{\mu}\in[-\pi,\pi)$. 

For a single layer, Eq.\,\eqref{eq:PhiDGL_general} has been studied numerically and analytically~\cite{AOK09,AOK11,NEK16,KAS16a,KAS17,BER19}. In particular, it was shown that starting from uniformly distributed random initial condition $\phi_{i}\in[0,2\pi)$, $\kappa_{ij}\in[-1,1]$ the system can reach different frequency multi-cluster states with hierarchical structure depending on the parameters $\alpha$ and $\beta$. The frequency multi-clusters in turn consist of several one-clusters which determine the existence and stability of the former~\cite{BER19a}. Therefore, these one-cluster states (with identical frequency, but different phase distributions) constitute the building blocks of adaptively coupled phase oscillators. 

Before we consider multiple layers, we will introduce notions and patterns of the single-layer case, which will be needed for the multiplex network. Each solution of Eq.\,\eqref{eq:PhiDGL_general} for $L=1$ or $L=2$ is called a monoplex or duplex state, respectively.

It is known that already the monoplex system \eqref{eq:PhiDGL_general} possesses a huge variety of dynamical states such as multi-clusters with respect to frequency synchronization, chaotic attractors, and chimera-like states~\cite{NEK16,KAS16a,KAS17,BER19, BER19a}. In this article we will focus on the simplest patterns, namely one-cluster states, and their generalization to the multiplex case. The reasons for this are threefold. First, from the analytical point of view, the one-cluster states in the monoplex network are very well understood. Second, the one-clusters are building blocks for more complex dynamical states such as multi-cluster states. Therefore, they are essential for the understanding of more complex dynamical patterns. And third, stable chimera-like states as they were studied in~\cite{KAS17,KAS18} exist close to the borders of the stability regions for one-clusters, so the existence and stability of one-clusters may pave the way for observing those hybrid patterns.

In general, one-cluster states are given by equilibria relative to a co-rotating frame~\cite{BER19} 
\begin{equation}
\label{eq:pPL}
    \begin{aligned}
	\phi_{i}^\mu &= \Omega t+a_{i}^\mu,\\
	\kappa_{ij}^\mu &= -\sin(a_{i}^\mu-a_{j}^\mu+\beta^{\mu}),
    \end{aligned}
\end{equation}
with collective frequency $\Omega$ and relative phases $a_{i}^\mu$. 
Hence the second moment order parameter $R_2(\mathbf{a}^{\mu}) = \frac1N\left|\sum_{j=1}^Ne^{\mathrm{i}2a_{j}^{\mu}}\right|$
with $\mathbf{a}^\mu\equiv(a_{1}^\mu,\dots,a_{N}^\mu)^T$ can be used as a characteristic measure. In the case of monoplex systems ($L=1$), three types of solutions exist (see Fig.~\ref{fig:1Cl_Illustration}) which are characterized by corresponding frequencies $\Omega$ as a function of $(\alpha^{11},\beta^{1})$ \cite{BER19}: (a) $\Omega=\cos(\alpha^{11}-\beta^{1})/2$ if $R_{2}(\mathbf{a}^1)=0$ (Splay state), (b) $\Omega=\sin\alpha^{11}\sin\beta^{1}$ if $R_{2}(\mathbf{a}^1)=1$ with $a_i^1\in\{0,\pi\}$ (Antipodal state), (c) $\Omega=\cos(\alpha^{11}-\beta^{1})/2-R_{2}(\mathbf{a})\cos(\psi_Q)/2$ if $0<R_{2}(\mathbf{a}^1)<1$ with $a_{i}^1\in\{ 0,\pi,\psi_{Q},\psi_{Q}+\pi\}$ (Double antipodal state) with $\psi_{Q}$ being the unique solution (modulo $2\pi$) of
\begin{align*}
	(1-q)\sin(\psi_{Q}-\alpha^{11}-\beta^{1})=q\sin(\psi_{Q}+\alpha^{11}+\beta^{1}),
\end{align*}
where $q=Q/N$ and $Q\in\left\{ 1,\dots,N-1\right\}$ denotes the number of phase shifts $a_{i}^1\in\{0,\pi\}$. Here, splay states are defined in a more general sense by $R_2(\mathbf{a}^1)=0$, which includes the states $a_{i}^1=2 \pi i/N$ usually referred to as splay state~\cite{CHO09}.
\begin{figure}
	\begin{center}
		\includegraphics[width=0.8\linewidth]{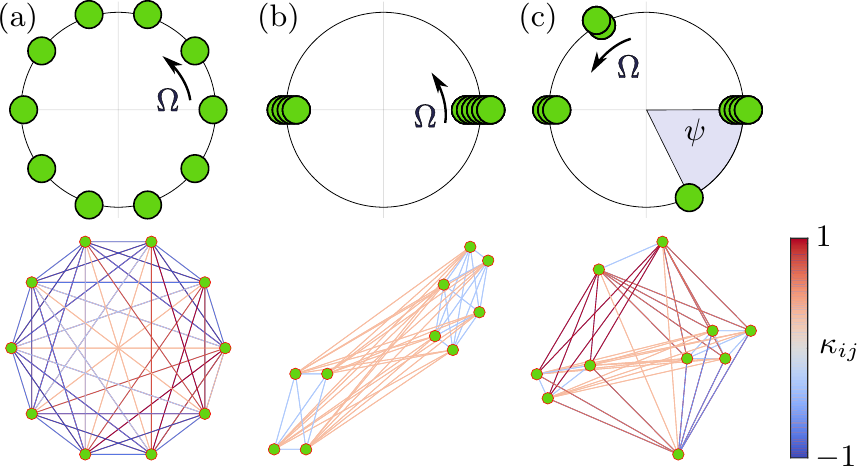}
	\end{center}%
	\caption{\label{fig:1Cl_Illustration}Illustration of the three types of monoplex one-cluster states of Eq.\,\eqref{eq:pPL} ($L=1$) for an ensemble of $10$ oscillators (green circles) with frequencies $\Omega$ (upper panels) and coupling structure with weights $\kappa_{ij}$ (lower panels): One-cluster (a) of splay type ($R_{2}(\mathbf{a})=0$), (b) of antipodal type ($R_{2}(\mathbf{a})=1$), and (c) of double antipodal type with $Q=7$. Parameters: $\alpha=0.1\pi$, $\beta=0.1\pi$}
\end{figure}

Let us now consider these one-cluster states in multiplex structures. Therefore, we introduce the notion of \emph{lifted} one-cluster states, where in each layer $\mu=1,\dots,L$ the state $(\phi^\mu_i(t), \kappa_{ij}^\mu(t))$ is a monoplex one-cluster, i.e., the phases $a_{i}^\mu$ of the oscillators are of splay, antipodal, or double antipodal type. It can be shown \cite{suppl} that in duplex systems ($L=2$) the phase difference of oscillators between the layers $\Delta a\equiv a_{i}^{1}-a_{i}^{2}$
takes only two values and solves $\Delta\Omega = \sigma^{12}\sin(\Delta a+\alpha^{12})+\sigma^{21}\sin(\Delta a-\alpha^{12})$,
where $\Delta\Omega\equiv\Omega(\alpha^{11},\beta^{1})- \Omega(\alpha^{22},\beta^{2})$ is given above for the three different one-cluster states (splay, antipodal, double antipodal). 
These states are shown in Fig.\,\ref{fig:Fig2_DuplexEqu}: Panels (a),(b),(d) display lifted states of splay, antipodal, and splay type, respectively. The phase distributions in both layers are the same but shifted by the constant value $\Delta a$ in agreement with the above equation. 
In contrast to the lifted states, Fig.\,\ref{fig:Fig2_DuplexEqu}(c) shows another possible one-cluster for the duplex network. Due to the interaction of the two layers we can find a phase distribution which is of double antipodal type in each layer but not a lifted state. This means that these states are born by the duplex set-up. Moreover, in contrast to the other examples the phase distribution between the layers does not agree, $\psi^1 \ne\psi^2$. For the monoplex case, it has been shown that double antipodal states are unstable for any set of parameters~\cite{BER19a}. Hence, finding stable double antipodal states which interact through the duplex structure is unexpected.
\begin{figure}
	\begin{center}
		\includegraphics[width=0.8\linewidth]{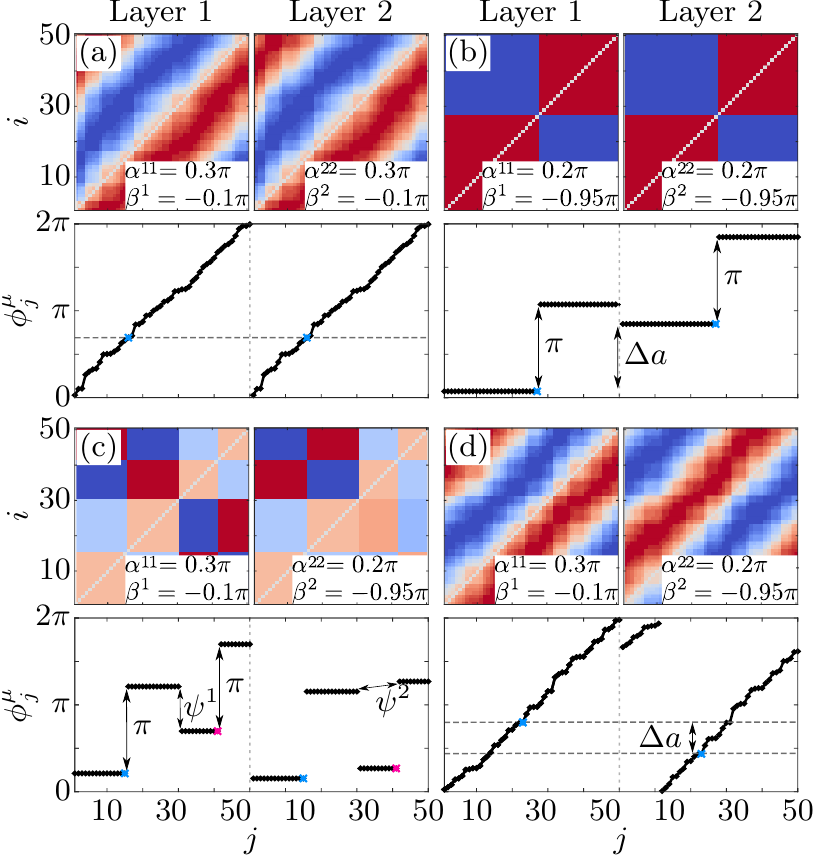}
	\end{center}%
	\caption{\label{fig:Fig2_DuplexEqu}Different duplex states of Eq.\,\eqref{eq:pPL} ($L=2$) for an ensemble of $50$ oscillators in each layer with color-coded coupling weights $\kappa_{ij}^{\mu}$ (upper panels, color code as in Fig.1), phases $\phi_j^{\mu}$ (lower panels): Duplex one-cluster states (a) of lifted splay type ($R_{2}(\mathbf{a^\mu})=0$) for $\alpha^{12/21} = 0.3\pi$, $\sigma^{12/21}=0.07$; (b) of lifted antipodal type ($R_{2}(\mathbf{a^\mu})=1$) for $\alpha^{12} = 0.3\pi$, $\alpha^{21}=0.75\pi$, $\sigma^{12/21}=0.62$; (c) of double antidodal type (not a lifted state) for $\alpha^{12/21}= 0.05\pi$, $\sigma^{12/21}=0.28$; (d) of lifted splay type for $\alpha^{12} = 0.3\pi$, $\alpha^{21}=0.4\pi$, $\sigma^{12/21}=0.8$, and $\epsilon=0.01$. In the lower panels phase differences between the two layers are indicated by $\Delta a\equiv a_{i}^{1}-a_{i}^{2}$, and between the two new antipodal states (c) by $\psi^1, \psi^2$.}
\end{figure}

For more insight into the birth of phase-locked states by multiplexing, Fig.\,\ref{fig:BirthDAP} displays the emergence of double antipodal states in a parameter regime where they do not exist in single-layer networks. They are characterized by the second moment order parameter $R_2$. It is remarkable that the new double antipodal state can be found for a wide range of the inter-layer coupling strength larger than a certain critical value $\sigma_c$, and is clearly different from those of the monoplex. Below the critical value $\sigma_c$, the double antipodal states are no longer stable, and more complex temporal dynamics occurs which causes temporal changes in $R_2$. This leads to non-vanishing temporal variance indicated by the error bars in Fig.~\ref{fig:BirthDAP}. 
\begin{figure}
	\begin{center}
		\includegraphics[width=0.8\linewidth]{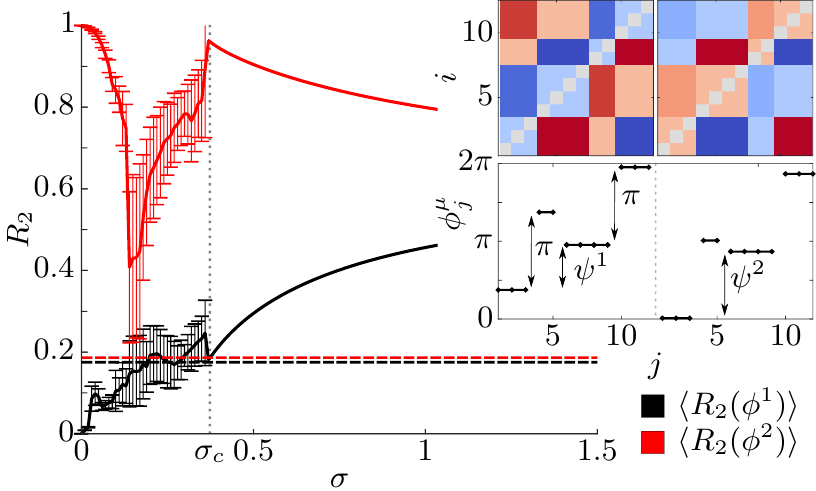}
	\end{center}%
	\caption{Birth of double antipodal state in a duplex network ($N=12$) for a wide range of inter-layer coupling strength $\sigma=\sigma^{12}=\sigma^{21}$. The solid lines are the temporal averages for the second moment order parameter $R_2$ of the individual layers (layer 1: black, layer 2: red). The error bars for $\sigma<\sigma_c$ denote the standard deviation of the temporal evolution of $R_2$. The dashed horizontal lines represent the unique values of $R_2$ for the double antipodal state in a monoplex network. The plot was obtained by adiabatic continuation of a duplex double antipodal state (see inset) in both directions starting from $\sigma=0.5$. Parameters: $\alpha^{11/22} = 0.3\pi$, $\alpha^{12/21}=0.05$, $\beta^{1}=0.1\pi$, $\beta^{2}=-0.95\pi$, and $\epsilon=0.01$.\label{fig:BirthDAP}}
\end{figure}

In the following we show how the dynamics in a neighborhood of theses states can be lifted as well, i.e., we investigate their local stability. Everything is exemplified for antipodal states but can be generalized to lifted splay and lifted double antipodal states in a straightforward manner. To study the dynamics around the one-cluster states described by Eq.\,\eqref{eq:pPL}, we linearize Eq.\,\eqref{eq:PhiDGL_general} around these states:
\begin{eqnarray}
\label{eq:Linearized_OneCl_phi}
&&\dot{\delta\phi}_{i}^\mu =\frac{1}{N}\sum_{j=1}^{N}\big[\sin(\Delta a+\beta^{\mu})\cos(\Delta a+\alpha^{\mu\mu})\Delta_{ij}^{\mu\mu} \delta\phi - \nonumber\\[-.35cm]
&&\sin(\Delta a+\alpha^{\mu\mu})\delta\kappa_{ij}^\mu\big]-\sum_{\nu=1}^M \sigma^{\mu\nu}\cos(\Delta a+\alpha^{\mu\nu})\Delta_{ij}^{\mu\nu} \delta\phi,\nonumber \\
&&\dot{\delta\kappa}_{ij}^\mu =-\epsilon\left(\delta\kappa_{ij}^\mu+\cos(\Delta a+\beta^{\mu})\Delta_{ij}^{\mu\mu} \delta\phi\right)
\end{eqnarray}
where $\Delta_{ij}^{\mu\nu} \delta\phi \equiv\delta\phi^\mu_i-\delta\phi^\nu_j$.

In duplex networks, the coupling structure is given by a $2\times 2$ block matrix $M$ with the $N\times N$ unity matrix $\mathbb{I}_N$:
 
	\begin{align}\label{eq:duplexForm}
		M & =\begin{pmatrix}A & m\cdot\mathbb{I}_N\\
			n\cdot\mathbb{I}_N & B
		\end{pmatrix}.
	\end{align}
If $A$ and $B$ are diagonalizable $N\times N$ matrices which commute ($m,n\in\mathbb{R}$, $n\ne 0$), the following relation for the characteristic polynomial can be proven \cite{suppl} using Schur's decomposition~\cite{BOY04,Liesen2015}:
	\begin{align}
	\label{lem:DuplexDecomp}
	\det\left(mn\cdot \mathbb{I}_N - (D_A-\lambda\mathbb{I}_N)(D_B-\lambda\mathbb{I}_N)\right)=0,
	\end{align}
where $D_A$ and $D_B$ are the diagonal matrices corresponding to $A$ and $B$, respectively. Note that Eq.\,\eqref{lem:DuplexDecomp} not only simplifies the calculation for the eigenvalues in the case of a duplex structure, moreover, it is a general result on linear dynamical systems on duplex networks. Therefore, this result is important for the investigation of stability and symmetry in multiplex networks.

In the case of a duplex antipodal one-cluster state Eq.\,\eqref{eq:PhiDGL_general} with $a^1_i\in \{0,\pi\}$ and $a_i^2=a_i^1-\Delta a$, Eq.\,\eqref{eq:Linearized_OneCl_phi} can be brought to the form~\eqref{eq:duplexForm} and possesses the following set of Lyapunov exponents
$\mathcal{S}=\{-\epsilon,\left(\lambda_{i,1},\lambda_{i,2},\lambda_{i,3},\lambda_{i,4}\right)_{i=1,\dots,N}\}$
where $\lambda_{i,1,\dots,4}$ are the solutions of polynomials containing the eigenvalues of the monoplex system \cite{suppl}.

Thus, the stability analysis of the duplex system can be reduced to that of the monoplex case. We are able to analyze the stabilizing and destabilizing features of a duplex network numerically and analytically. To illustrate the effect of multiplexing, the interaction between two clusters of antipodal type is presented in Fig.\,\ref{fig:DuplexFeat}. The stability of these states is determined by integrating Eq.\,\eqref{eq:pPL} numerically starting with a slightly perturbed lifted antipodal state. The states are stable if the numerical trajectory is approaching the lifted antipodal state. Otherwise, the state is considered as unstable. The black contour lines in Fig.\,\ref{fig:DuplexFeat} show the borders of the stability regions in dependence of the coupling strength $\sigma^{21}$, as calculated from the Lyapunov exponents.
The borders are in remarkable agreement with the numerical results.
\begin{figure}
	\begin{center}
		\includegraphics[width=0.8\linewidth]{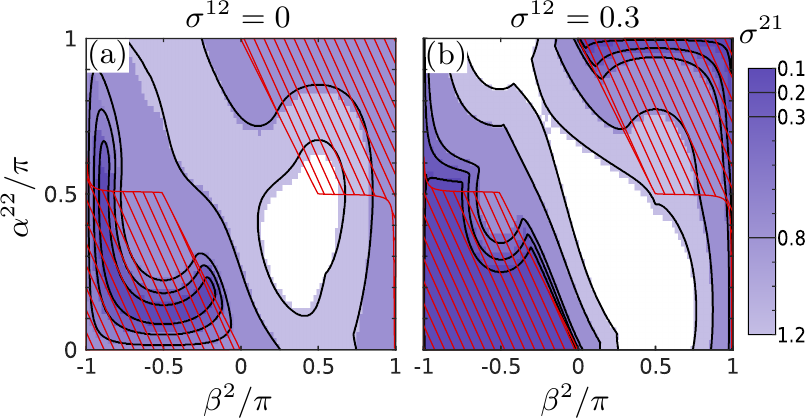}
	\end{center}%
	\caption{Regions of stability (blue) and instability (white) of the lifted antipodal state in the $(\alpha^{22},\beta^{2})$ parameter plane for different values of interlayer coupling (indicated by different blue shading) $\sigma^{21}$, where regions of stronger coupling $\sigma^{21}$ (lighter blue) include such of weaker $\sigma^{21}$ (darker blue). Stability regions for single-layer antipodal clusters are indicated by red hatched areas. The inter-layer coupling is considered as (a) unidirectional ($\sigma^{12}=0$) and (b) bidirectional ($\sigma^{12}=\sigma^{21}$). Parameters: $\alpha^{11} = 0.2\pi$, $\beta^{1}=-0.8\pi$, $\alpha^{12}=0$, $\alpha^{21}=0.3\pi$, and $\epsilon=0.01$.\label{fig:DuplexFeat}}
\end{figure}

We investigate two different situations in Fig.\,\ref{fig:DuplexFeat}: In both panels the parameters for the first layer $\alpha^{11}, \beta^{1}$ are chosen such that the antipodal state is stable without inter-layer coupling. The stability of the duplex antipodal states is displayed in the $(\alpha^{22},\beta^{2})$ parameter plane for several values of the inter-layer coupling $\sigma^{21}$. To compare the effects of the duplex network with the mono-layer case, the stability regions for monoplex antipodals states are displayed, as well, as red hatched areas. They are markedly different. In Figure \ref{fig:DuplexFeat}(a), the two layers are connected unidirectionally ($\sigma^{12}=0$). It can be seen that with increasing inter-layer coupling weight $\sigma^{21}$ the region of stability for the lifted antipodal state also grows. Already for small values of the inter-layer couplings $\sigma^{21}$, a stabilizing effect of the duplex network can be noticed. For $\sigma=0.1$ there exist already regions for which the duplex antipodal state is stable but the corresponding monoplex state would not be stable. The opposite effect is found as well where the duplex network destabilizes a lifted state. Figure~\ref{fig:DuplexFeat}(b) shows the results for two layers with bidirectional coupling. In this case, the duplex structure can have stabilizing and destabilizing effects, as well. Further, for the bidirectional coupling we also notice a growth of the stability region with increasing $\sigma^{21}$ similar to the unidirectional case. However, the regions of stability grow at different rates in dependence on $\sigma^{21}$ and non-monotonically with respect to the parameters $\alpha^{22},\beta^{2}$. Comparing the size of the stability region for both cases, one can see that for small values of $\sigma^{21}$ the region for bidirectional coupling is larger. In turn, for higher inter-layer coupling, the regions for the unidirectional case are larger. It is worth noting that in Fig.\,\ref{fig:DuplexFeat} the stability regions for smaller values of $\sigma^{21}$ are always contained in the region for larger values of $\sigma^{21}$.

In conclusion, we have proposed a concept to induce diverse partial synchronization patterns (phase clusters) in adaptively coupled phase oscillator networks.  While adaptive networks have recently attracted a lot of attention in the fields of neuro- and social sciences, biology, engineering, and other disciplines, and multilayer networks are a paradigm for real-world complex networks, little has been known about the interplay of multilayer structures and adaptivity. We have aimed to fill this gap within a rigorous framework of theoretical analysis and computer simulations.  We have shown that multiplexing in a multi-layer with symmetry can induce various stable phase cluster states like splay states, antipodal states, and double antipodal states, in a situation where they are not stable or do not even exist in the single layer. Further, we have developed a novel method for analysis of Laplacian matrices of duplex networks which allows for insight into the spectral structure of these networks, and can easily be generalized to more than two layers~\cite{suppl}. This new approach of multiplex decomposition has a broad range of applications to physical, biological, socio-economic, and technological systems, ranging from plasticity in neurodynamics or the dynamics of linear diffusive systems~\cite{GOM13,Sole-Ribalta2013} to generalizations of the master stability approach~\cite{PEC98,TAN19} for adaptive networks~\cite{suppl}. We have used the multiplex decomposition to provide analytic results for the stability of lifted states in the multilayer system. As local dynamics we have used the paradigmatic Kuramoto phase oscillator model, supplemented by adaptivity of the link strengths with a phase lag parameter which can model a whole range of adaptivity rules from Hebbian via spike-timing dependent plasticity to anti-Hebbian.

\begin{acknowledgments}
    This work was supported by the German Research Foundation DFG (Projects SCHO 307/15-1 and YA 225/3-1 and Projektnummer - 163436311 - SFB 910). We thank Serhiy Yanchuk for insightful discussions.
\end{acknowledgments}


%

\clearpage

\begin{center}
	\textbf{\large }
\end{center}
\pagebreak

\onecolumngrid
\begin{center}
	\textbf{\large Supplemental Material on\\ Birth and stabilization of phase clusters by multiplexing of adaptive networks}\\[.2cm]
	Rico Berner$^{1,2,*}$, Jakub Sawicki$^{1}$, and Eckehard Sch\"oll$^{1,\dagger}$\\[.1cm]
	{\itshape $^{1}$Institut f\"ur Theoretische Physik, Technische Universit\"at Berlin, Hardenbergstr.\,36, 10623 Berlin, Germany\\
		$^{2}$Institut f\"ur Mathematik, Technische Universit\"at Berlin, Hardenbergstr.\,36, 10623 Berlin, Germany\\}
\end{center}
\twocolumngrid

\setcounter{equation}{0}
\setcounter{figure}{0}
\setcounter{table}{0}
\setcounter{page}{1}
\renewcommand{\theequation}{S\arabic{equation}}
\renewcommand{\thefigure}{S\arabic{figure}}
\renewcommand{\bibnumfmt}[1]{[S#1]}
\renewcommand{\citenumfont}[1]{S#1}





\section{Model}\label{sec:Intro}
We consider a multiplex network with $L$ layers each consisting of $N$ identical adaptively coupled phase oscillators
\begin{align}
	\frac{d\phi_{i}^{\mu}}{dt} &=\omega-\frac{1}{N}\sum_{j=1}^{N}\kappa_{ij}^\mu\sin(\phi_{i}^\mu-\phi_{j}^\mu+\alpha^{\mu\mu}) \label{eq:PhiDGL_general}\\
	&\phantom{=\omega} -\sum_{\nu=1,\nu\ne\mu}^L\sigma^{\mu\nu}\sin(\phi_{i}^\mu-\phi_{i}^\nu+\alpha^{\mu\nu}), \nonumber \\
	\frac{d\kappa_{ij}^\mu}{dt}&=-\epsilon\left(\kappa_{ij}^\mu+\sin(\phi_{i}^\mu-\phi_{j}^\nu+\beta^{\mu})\right),\label{eq:KappaDGL_general}
\end{align}
where $\phi_{i}^{\mu}\in [0,2\pi)$ represents the phase of the $i$th
oscillator ($i=1,\dots,N$) in the $\mu$th layer ($\mu=1,\dots,L$) and $\omega$ is the natural frequency. The
interaction between the phase oscillators within each layer is described by the coupling
matrix $\kappa_{ij}^\mu\in[-1,1]$. The intra-layer coupling weights obey equation~(\ref{eq:KappaDGL_general}). Between the layers the interaction is given by the fixed coupling weights $\sigma^{\mu\nu} \ge 0$. The parameters $\alpha^{\mu\nu}$ can be considered as a phase lag of the interaction~\cite{SAK86s}.

In order to include also more general dynamics on static networks we also consider the following dynamical equations with diffusive coupling~\cite{Keane2012s,Lehnert2016s,Pecora1998s}
\begin{multline}\label{eq:CoupledSysDiff}
	\dot{\mathbf{x}}^{k}_i=f(\mathbf{x}^{k}_i(t))-\sigma\sum_{j=1}^N a_{ij}^{k}H(\mathbf{x}^{k}_i-\mathbf{x}^{k}_j)\\
	-\rho\sum_{l=1}^L m^{kl}G(\mathbf{x}^{k}_i-\mathbf{x}^{l}_i),
\end{multline}
where $\mathbf{x}^{k}_i$ is a $d$-dimensional state vector of the $i$th node in the $k$th layer, \textit{i.e.} $\mathbf{x}^{k}_i\in\mathbb{C}^d$, $f\in\mathcal{C}^1(\mathbb{C}^d,\mathbb{C}^d)$ describes the local dynamics, the functions ${H},{G}\in\mathcal{C}^1(\mathbb{C}^d,\mathbb{C}^d)$ determine the intra- and inter-layer coupling scheme, respectively. The parameter $\sigma\in\mathbb{R}$ is the intra-layer coupling strength, $\rho\in\mathbb{R}$ and is the inter-layer coupling strength. The connectivities are given by the matrix elements ${a^{k}_{ij}\in\{0,1\}}$ of the $N\times N$ intra-layer adjacency matrix $A^{k}$ and the $L\times L$ (inter)-layer adjacency matrix $M$ with elements $m^{kl}\in\{0,1\}$.
\section{Existence of duplex equilibria in adaptive networks}\label{sec:ExistenceDuplexEqui}
Suppose we have two one-cluster states where each is of either splay, antipodal, or double antipodal type which form a duplex one-cluster (see Eq.~(2) of the main text) and  $\phi_{i}^\mu=\Omega(\alpha^{\mu\mu},\beta^\mu)t+\omega^\mu t+a_{i}^\mu$ ($\mu=1,2$), where $\Omega(\alpha^{\mu\mu},\beta^\mu)$ is given by Eq.~(3), of the main text, and the coupling weights are given by $\kappa_{ij}^\mu=-\sin(a_{i}^\mu-a_{j}^\mu+\beta^{\mu})$. We verify by directly inserting that $\phi_{i}^\mu$ and $\kappa_{ij}^\mu$ solve Eq.~\eqref{eq:KappaDGL_general}. For the given ansatz, Eq.~\eqref{eq:PhiDGL_general} reads
\begin{multline*}
	\Omega(\alpha^{11},\beta^1)+\omega^1\\ =\frac{1}{2}\cos(\alpha^{11}-\beta^{1})-\frac{1}{2}\Re\left(e^{-\mathrm{i}(2a_{i}^1+\alpha^{11}+\beta^{1})}Z_{2}(\mathbf{a}^1)\right)\\
	-\sigma^{12}\sin(\Delta\Omega t+\Delta\omega t + a_{i}^1-a_{i}^2+\alpha^{12}),
\end{multline*}
and
\begin{multline*}
	\Omega(\alpha^{22},\beta_2)+\omega^2\\ =\frac{1}{2}\cos(\alpha^{22}-\beta^{2})-\frac{1}{2}\Re\left(e^{-\mathrm{i}(2a_{i}^1+\alpha^{22}+\beta^{2})}Z_{2}(\mathbf{a}^2)\right)\\
	+\sigma^{21}\sin(\Delta\Omega t+\Delta\omega t+a_{i}^1-a_{i}^2-\alpha^{21}).
\end{multline*}
where $\Delta\Omega=\Omega(\alpha^{11},\beta^1)-\Omega(\alpha^{22},\beta^2)$ and $\Delta\omega=\omega^1-\omega^2$, respectively.
Thus, $\phi=(\phi^1,\phi^2)$ is a duplex equilibrium if
\begin{align*}
	\Delta\Omega + \Delta\omega = 0
\end{align*}
which is equivalent to
\begin{multline*}
	\Delta\Omega = \sigma^{12}\sin(a_{i}^1-a_{i}^2+\alpha^{12})+\\ 
	\sigma^{21}\sin(a_{i}^1-a_{i}^2-\alpha^{21})
\end{multline*}
for all $i=1,\dots,N$. Note that $\Delta\Omega$ is not necessarily zero even if the phase-lag parameters for both layers agree. They can still differ in the type of one-cluster state. The former equation can be written as
\begin{align}\label{eq:Double1Cl_Existence}
	\frac{\Delta\Omega}{C} = \sin(a_{i}^1-a_{i}^2+\nu)
\end{align}
with
\begin{align}
	\sin(\nu)&=\frac{1}{C}\left(\sigma^{12}\sin(\alpha^{12})-\sigma^{21}\sin(\alpha^{21})\right), \label{eq:sinnu_double1Cl}\\
	\cos(\nu)&=\frac{1}{C}\left(\sigma^{12}\cos(\alpha^{12})+\sigma^{21}\cos(\alpha^{21})\right), \nonumber
\end{align}
where
\begin{align*}
	C=\sqrt{(\sigma^{12})^2+(\sigma^{21})^2+2\sigma^{12}\sigma^{21}\cos(\alpha^{12}+\alpha^{21})}.
\end{align*}
Whenever $(\sigma^{12})^2+(\sigma^{21})^2+2\sigma^{12}\sigma^{21}\cos(\alpha^{12}+\alpha^{21})\ge 0$ and
\begin{align}\label{eq:sigmaduplexrel}
	(\sigma^{12})^2+(\sigma^{21})^2+2\sigma^{12}\sigma^{21}\cos(\alpha^{12}+\alpha^{21})\ge\Delta\Omega^2,
\end{align}
Eq.~(\ref{eq:Double1Cl_Existence}) has the two solutions $a_{i}^1-a_{i}^2=\arcsin(\Delta\Omega/C)-\nu$ and $a_{i}^1-a_{i}^2=\pi-\arcsin(\Delta\Omega/C)-\nu$. Considering the inverse function $\arcsin:[-1,1]\to[-\pi/2,\pi/2]$ applied to Eq.~(\ref{eq:sinnu_double1Cl}) determines $\nu$ to be either $\nu'$
or $\pi-\nu'$, where $\nu':=\arcsin(\sin(\nu))$ and $\sin(\nu)$ as given in (\ref{eq:sinnu_double1Cl}). The second equation for $\cos(\nu)$ then fixes $\nu$ to take one of the values.

The condition~\eqref{eq:sigmaduplexrel} is a relation between all parameters of the system which has to be fulfilled for the existence of duplex relative equilibria. Note that for any given inter-layer coupling $\sigma^{12}\ne 0$ and $\alpha^{12}+\alpha^{21}\ne\pm\pi/2$ or $\pm 3\pi/2$ there exists a minimum coupling weight $\sigma^{21}<\infty$ such that the lifted one-clusters exist. In case of unidirectional coupling, i.e., $\sigma^{12}=0$, the condition gives the minimum weight $\sigma^{21}\ge \Delta\Omega$.
\section{Multiplex networks and their decomposition}
In this section, we provide important tools and theorems to find the spectrum of multiplex networks.
\begin{thm}\label{thm:DeterminantBlockMatrixOnRing}
	Let $\mathcal{R}$ be a commutative subring of $\mathbb{C}^{N\times N}$ and let $M\in \mathcal{R}^{L\times L}$. Then,
	\begin{align*}
		{\det}_\mathbb{C}M = {\det}_\mathbb{C} \left({\det}_{\mathcal{R}} M\right). 
	\end{align*}
\end{thm}
The proof can be found in Ref~\onlinecite{Silvester2000}. This rather abstract result allows for a very nice decomposition for pairwise commuting matrices and yields a useful tool to study the local dynamics in multiplex systems.
\begin{prop}\label{prop:DeterminantDecompDiagAble}
	Let $M\in \mathbb{C}^{N\times N}$ be a unitary diagonalizable matrix with $M=U D_M U^{H}$ where $U$, $U^H$ and $D_M$ are a unitary, its adjoint and a diagonal matrix, respectively. Let further $\mathcal{D}_M$ be the set of simultaneously diagonalizable matrices to $M$, i.e., the set of all matrices which commute pairwise and with $M$. Then,
	\begin{multline}\label{eq:DeterminantDecompDiagAble}
		\det \begin{pmatrix}
			A_{11} & \cdots & A_{1L}\\
			\vdots & \ddots & \vdots\\
			A_{L1} & \cdots & A_{LL}
		\end{pmatrix}=\\
		\det\left(\sum_{\sigma\in S_L}\left[\textrm{sgn}(\sigma)\prod_{\mu=1}^{L}D_{A_{\mu,\sigma(\mu)}}\right]\right)
	\end{multline}
	where $A_{\mu\nu}\in\mathcal{D}_M$ for $\mu,\nu=1,\dots,L$ and $S_L$ is the set of all permutations of the numbers $1,\dots,L$.
\end{prop}
\begin{proof}
	Consider any $A,B \in \mathcal{D}_M$, then they are simultaneously diagonalizable with $M$ and hence $A=D_A U^H$ and $B=UD_B U^H$ with the same $U$. Thus, all $A_{\mu\nu}$ can be diagonalized with the same $U$. Since $U$ is unitary,\textit{i.e.} $(\det U)^2 = 1$, we find
	\begin{multline*}
		\det \begin{pmatrix}
			A_{11} & \cdots & A_{1L}\\
			\vdots & \ddots & \vdots\\
			A_{L1} & \cdots & A_{LL}
		\end{pmatrix}
		=\det \begin{pmatrix}
			D_{A_{11}} & \cdots & D_{A_{1L}}\\
			\vdots & \ddots & \vdots\\
			D_{A_{L1}} & \cdots & D_{A_{LL}}
		\end{pmatrix}
	\end{multline*}
	by applying the block diagonal matrices $\text{diag}(U,\cdots,U)$ and $\text{diag}(U^H,\cdots,U^H)$ from the left and right, respectively. The set of diagonal matrices with usual matrix multiplication and addition form a commutative subring of $\mathbb{C}^{N\times N}$. Applying Theorem~\ref{thm:DeterminantBlockMatrixOnRing} and using the well-known determinant representation of Leibniz, the expression~\eqref{eq:DeterminantDecompDiagAble} follows.
\end{proof}
\begin{rem}
	The set $\mathcal{D}_M$ consists of all matrices which commute with $M$ and all the other elements of $\mathcal{D}_M$. In particular, the identity matrix $\mathbb{I}_N\in\mathcal{D}_M$ for any $M\in\mathbb{C}^{N\times N}$.
\end{rem}
In the following, we apply the last result to a duplex and triplex system and connect the local dynamics on the one-layer network to the multiplex case. We specify our consideration by defining two special multiplex systems.
\begin{defn}\label{prop:ecompDiagAble}
	Suppose $A,B,C\in\mathbb{C}^{N\times N}$ and $m_{ij}\in\mathbb{C}$ ($i,j=1,\dots,3$). Then, the $2N\times2N$ block matrix
	\begin{align}\label{eq:DuplexAdj}
		M^{(2)} = \begin{pmatrix}
			A & m_{12}\mathbb{I}\\
			m_{21}\mathbb{I} & B
		\end{pmatrix} 
	\end{align}
	and the $3N\times3N$ block matrix
	\begin{align}\label{eq:TriplexAdj}
		M^{(3)} = \begin{pmatrix}
			A & m_{12}\mathbb{I} & m_{13}\mathbb{I}\\
			m_{21}\mathbb{I} & B & m_{23}\mathbb{I}\\
			m_{31}\mathbb{I} & m_{32}\mathbb{I} & C
		\end{pmatrix} 
	\end{align}
	are called (complex) duplex and triplex network, respectively.
\end{defn}
Suppose we know how to diagonalize the individual layer topologies. The next result shows how the eigenvalues of the individual layers are connected to eigenvalues of the multiplex system. This will be done for the duplex and triplex network. For the proof of our following statement, we provide two different proofs. The first approach makes use of Schur's decomposition~\cite{BOY04s,Liesen2015s} which will be used, later on, in order to derive the characteristic equations. In particular, any $m\times m$ matrix $M=\begin{pmatrix}A & B\\C & D\end{pmatrix}$ in the $2\times2$ block form can be written as
\begin{align}\label{eq:SchurComplement}
	M
	=\begin{pmatrix}\mathbb{I}_{p} & BD^{-1}\\
		0 & \mathbb{I}_{q}
	\end{pmatrix}\begin{pmatrix}A-BD^{-1}C & 0\\
		0 & D
	\end{pmatrix}\begin{pmatrix}\mathbb{I}_{p} & 0\\
		D^{-1}C & \mathbb{I}_{q}
	\end{pmatrix}.
\end{align}
With this, a simplified form of the determinant of a $2\times2$ block matrix $M=\begin{pmatrix}A & B\\C & D\end{pmatrix}$ is derived, namely
\begin{align}
	\det(M) & =\det(A-BD^{-1}C)\cdot\det(D),
\end{align}
An extension of the first approach to any number of layers in the network can be found by induction but is very technical, see~\cite{Powell2011s}. The second approach uses Proposition~\ref{prop:DeterminantDecompDiagAble} which allows for a straightforward extension to any number of layers in a multiplex network.
\begin{prop}\label{prop:EigenvaluePolynomsMultiplex}
	Suppose $A,B,C\in\mathbb{C}^{N\times N}$, they commute pairwise, and are diagonalizable with diagonal matrices $D_A,D_B,D_C$ and unitary matrix $U$. Then, the eigenvalues $\mu$ for the multiplex networks $M^{(2)}$ and $M^{(3)}$ can be found by solving the $N$ quadratic
	\begin{align}\label{eq:DuplexDecomp}
		\mu^2-\left((d_A)_i+(d_B)_i\right)\mu+(d_A)_i(d_B)_i-m_{12}m_{21} = 0
	\end{align}
	and cubic polynomial equations
	\begin{align}\label{eq:TriplexDecomp}
		\mu^3+a_{2,i}\mu^2+a_{1,i}\mu+a_{0,i}= 0,
	\end{align}
	respectively, with
	\begin{align*}
		a_{2,i} &= -\left((d_A)_i+(d_B)_i+(d_C)_i\right) \\
		a_{1,i} &= (d_A)_i(d_B)_i + (d_A)_i(d_C)_i + (d_B)_i(d_D)_i \\
		&- m_{12}m_{21} - m_{13}m_{31} - m_{23}m_{32} \\
		a_{0,i} &= m_{12}m_{21}(d_C)_i + m_{13}m_{31}(d_B)_i + m_{23}m_{32}(d_A)_i \\
		&- (d_A)_i(d_B)_i(d_C)_i - m_{12}m_{23}m_{31} -m_{13}m_{32}m_{21}
	\end{align*}
	and $(d_A)_i,(d_B)_i$, and $(d_C)_i$ the respective diagonal elements of $D_A, D_B$, and $D_C$. 
\end{prop}
\begin{proof}
	Since $A,B,C$ are diagonalizable and commute, Proposition~\ref{prop:DeterminantDecompDiagAble} can be applied to both matrices $M^{(2)}, M^{(3)}$. Anyhow, for the matrix $M^{(2)}$ we will provide another proof using Schur's decomposition.
	
	The determinant is an antisymmetric multi-linear form. Thus, we can write
	\begin{multline*}
		\det\left(M^{(2)}-\mu\mathbb{I}_{2N}\right) =
		\det\begin{pmatrix}A-\mu\mathbb{I}_N & m_{12}\cdot\mathbb{I}_N\\
			m_{21}\cdot\mathbb{I}_N & B-\mu\mathbb{I}_Nm 
		\end{pmatrix}\\
		= (-1)^N\det\begin{pmatrix}m_{12}\cdot\mathbb{I}_N & A-\mu\mathbb{I}_N\\
			B-\mu\mathbb{I}_N & m_{21}\cdot\mathbb{I}_N.
		\end{pmatrix}
	\end{multline*}
	By assumption $A$ and $B$ are both diagonalizable with respect to the unitary transformation matrix $U$, and so are $A-\mu\mathbb{I}$ and $B-\mu\mathbb{I}$. This allows us to write
	\begin{multline*}
		\det\begin{pmatrix}
			m_{12}\cdot\mathbb{I}_N & A-\mu\mathbb{I}_N\\
			B-\mu\mathbb{I}_N & m_{21}\cdot\mathbb{I}_N.
		\end{pmatrix}\\
		= \det
		\begin{pmatrix}
			m_{12}\mathbb{I}_N & D_A-\mu\mathbb{I}_N\\
			D_B-\mu\mathbb{I}_N & m_{21}\mathbb{I}_N
		\end{pmatrix}
	\end{multline*}
	we apply the block diagonal matrices $\text{diag}(U,\cdots,U)$ and $\text{diag}(U^H,\cdots,U^H)$ from the left and right, respectively. Now, using Schur's decomposition~\eqref{eq:SchurComplement} the determinant can written as
	\begin{multline*}
		\det
		\begin{pmatrix}
			m_{12}\mathbb{I}_N & D_A-\mu\mathbb{I}_N\\
			D_B-\mu\mathbb{I}_N & m_{21}\mathbb{I}_N
		\end{pmatrix}\\
		=n^N\det\left(m-\frac{1}{n}\left(D_A-\mu\mathbb{I}_N\right)\left(D_B-\mu\mathbb{I}_N\right)\right)\\
		= \det\left(m_{12}m_{21}\mathbb{I}_N-\left(D_A-\mu\mathbb{I}_N\right)\left(D_B-\mu\mathbb{I}_N\right)\right).
	\end{multline*}
	The last expression together with $\det\left(M^{(2)}-\mu\mathbb{I}_{2N}\right)=0$ yields the $N$ quadratic equations~\eqref{eq:DuplexDecomp}.
	
	Using that $(A-\mu\mathbb{I}),(B-\mu\mathbb{I}),(C-\mu\mathbb{I})$ commute pairwise, Proposition~\ref{prop:DeterminantDecompDiagAble} can be applied.  We find
	\begin{multline*}
		\det\left(M^{(3)}-\mu\mathbb{I}_{3N}\right)=\\
		\det\left((D_A-\mu\mathbb{I}_N)\left[(D_B-\mu\mathbb{I}_N)(D_C-\mu\mathbb{I}_N)-m_{23}m_{32}\mathbb{I}_N\right]\right.\\
		\left. -m_{21}\left[m_{12}(D_C-\mu\mathbb{I}_N)-m_{13}m_{32}\mathbb{I}_N\right]\right.\\
		\left. +m_{31}\left[m_{12}m_{23}\mathbb{I}_N - m_{13}(D_B-\mu\mathbb{I}_N)\right]\right)
	\end{multline*}
	The last expression together with $\det\left(M^{(3)}-\mu\mathbb{I}_{3N}\right)=0$ yields the $N$ cubic equations~\eqref{eq:TriplexDecomp}.
\end{proof}
Let us briefly discuss some special cases for both the duplex and triplex network. Consider a duplex network with master and slave layer, \textit{i.e.}, either $m_{12}=0$ or $m_{21}=0$. Then, the quadratic equations~\eqref{eq:DuplexDecomp} yield 
\begin{align}\label{eq:DuplexDecMasterSlave}
	\left(\mu-(d_A)_i\right)\left(\mu-(d_B)_i\right) = 0.
\end{align}
As shown in Proposition~\ref{prop:DeterminantDecompDiagAble}, the eigenvalues for special triplex networks can be found by solving cubic equations. For the solution even closed forms exist. Despite this, the explicit form of the solutions is rather tedious, in general. However, if we consider $A=B=C$ and a ring-like inter-layer connection between the networks, \textit{i.e.}, $m_{12}=m_{23}=m_{31}=0$, then equation~\eqref{eq:TriplexDecomp} has the following solutions for all $j=1,\dots,N$
\begin{align*}
	\mu_1 &= -(d_A)_j+\left(m_{13}m_{32}m_{21}\right)^{1/3},\\
	\mu_2 &= -(d_A)_j+\frac12 \mathrm{i}(\mathrm{i}+\sqrt{3})\left(m_{13}m_{32}m_{21}\right)^{1/3},\\
	\mu_3 &= -(d_A)_j-\frac12 (\mathrm{i}+\sqrt{3})\left(m_{13}m_{32}m_{21}\right)^{1/3},\\
\end{align*}
where $\mathrm{i}$ denotes the imaginary unit. In analogy to equation~\eqref{eq:DuplexDecMasterSlave}, a decoupling for the eigenvalues can be found. Consider three pairwise commuting matrices $A,B,C$, and the structure between the layers is a directed chain, \text{i.e.}, $m_{12}=m_{13}=m_{31}=m_{23}=0$ , then
\begin{align}\label{eq:TriplexDecMasterSlave}
	\left(\mu-(d_A)_i\right)\left(\mu-(d_B)_i\right)\left(\mu-(d_C)_i\right) = 0.
\end{align}
In the following three section, different applications will be briefly discussed.
\subsection{The master stability approach for multiplex networks}\label{sec:MSA_MltPlx}
In Ref.~\cite{Tang2019s}, the master stability function for dynamical systems on multiplex networks was introduced. In their article, the authors considered a systems similar to~\eqref{eq:CoupledSysDiff} but with Laplacian instead of adjacency matrices. Since they assumed the coupling function ${H}$ and ${G}$ to be linear, both systems are equivalent. Consider now the synchronous solution which solves $\dot{\mathbf{s}}=f(\mathbf{s})$. The master stability function corresponding to~\eqref{eq:CoupledSysDiff} can be derived from the following variational equations
\begin{align*}
	\dot{\bm{\xi}}=\left[\mathbb{I}_{NL}\otimes Df(\mathbf{s})-\sigma\left(\mathcal{L}^{\text{intra}}\otimes H\right)- \rho \left(\mathcal{L}^{\text{inter}}\otimes G\right)\right]\bm{\xi},
\end{align*} 
all necessary details can be found in Ref.~\cite{Tang2019s}. Here, $\bm{\xi}=\mathbf{x}-\mathbb{I}_L\otimes\mathbb{I}_N\otimes\mathbf{s}$ and $\mathbf{x}\in\mathbb{C}^{L\cdot N\cdot d}$ is the system state vector where all individual nodal states are stacked on each other ordered by the layer and node index. Further, the intra-layer Laplacian is defined is defined as
\begin{align*}
	\mathcal{L}^{\text{intra}}=\bigoplus_{l=1}^L L^{l}=\begin{pmatrix}
		L^{1}& &\\
		& \ddots&\\
		& & L^{L}
	\end{pmatrix}
\end{align*}
where
\begin{align*}
	L^{k}=\begin{pmatrix}
		\sum_{j=1}^N a^{k}_{1j} & &\\
		& \ddots &\\
		& & \sum_{j=1}^N a^{k}_{Nj}
	\end{pmatrix}-A^{k}.
\end{align*}
The inter-layer Laplacian is defined as $\mathcal{L}^{\text{inter}}=L^{I}\otimes \mathbb{I}_N$ where
\begin{align*}
	L^{I}=\begin{pmatrix}
		\sum_{l=1}^L m^{1l} & &\\
		& \ddots &\\
		& & \sum_{l=1}^L m^{Ll}
	\end{pmatrix}-M,
\end{align*}
where $m^{kl}$ with $m^{ll}=0$ are the entries of the $L\times L$ matrix $M$.
In Ref~\cite{Tang2019s}, it is shown that if $\mathcal{L}^{\text{intra}}$ and $\mathcal{L}^{\text{inter}}$ commute, a master stability equation for system~\eqref{eq:CoupledSysDiff} can be found which reads
\begin{align*}
	\dot{\mathbf{y}}=\left[Df(\mathbf{s})-\alpha H-\beta G\right]\mathbf{y},
\end{align*}
where $\alpha=\sigma\lambda$, $\beta=\rho\mu$, $\lambda$ and $\mu$ are the (complex) eigenvalues of $\mathcal{L}^{\text{intra}}$ and $\mathcal{L}^{\text{inter}}$, respectively. Note that if $H=G$ the master stability equation can be reduced to
\begin{align}\label{eq:MSEComposite}
	\dot{\mathbf{y}}=\left[Df(\mathbf{s})-\gamma H\right]\mathbf{y},
\end{align}
with $\gamma=\alpha+\beta$ and $\mathbf{y}\in\mathbb{C}^d$. Equation~\eqref{eq:MSEComposite} is called the master stability equation for the composite system where a single supra-Laplacian matrix~$\sigma\mathcal{L}^{\text{intra}}+\rho\mathcal{L}^{\text{inter}}$ described the network topology~\cite{DeDomenico2013s,Kivela2014s}.

Direct evaluation shows that $\left[\mathcal{L}^{\text{intra}},\mathcal{L}^{\text{inter}}\right]=0$ is equivalent to $m^{kl}L^{l}-m^{lk}L^{k}=0$ for all $l,k=1,\dots,L$. Thus, to require $\left[\mathcal{L}^{\text{intra}},\mathcal{L}^{\text{inter}}\right]=0$ yields a linear dependence between the individual layer topologies. Using Proposition~\ref{prop:ecompDiagAble} for composite system, the master stability function can be used under much milder conditions.
\begin{prop}\label{prop:MSE_general}
	Let us consider the multiplex dynamical system~\eqref{eq:CoupledSysDiff}, with linear function $H=G$. Suppose that further all $L^{l}$, $l=1,\dots,L$, commute pairwise. Then, the master stability equation is given by
	\begin{align}\label{eq:MSEComposite_mild}
		\dot{\mathbf{y}}=\left[Df(\mathbf{s})-\mu H\right]\mathbf{y},
	\end{align}
	and $\mu=\mu(\lambda^{1}_i,\dots,\lambda^{l}_i)$ being a non-linear function, mapping the $i$th eigenvalues $\lambda^{l}_i$ ($i=1,\dots,N$) of the layer topologies $L^{l}$, to the "master parameter" $\mu$ in~\eqref{eq:MSEComposite_mild}. The non linear mapping is given as the formal solution to the $L$th order polynomial equation
	\begin{align}
		\det\left(\sum_{\sigma\in S_L}\left[\textrm{sgn}(\sigma)\prod_{l=1}^{L}D_{\left(\mathcal{L}^{\text{supra}}_{l\sigma(l)}-\mu\mathbb{I}_N\delta_{l\sigma(l)}\right)}\right]\right)=0,
	\end{align}
	where $\mathcal{L}^{\text{supra}}=\sigma\mathcal{L}^{\text{intra}}+\rho\mathcal{L}^{\text{inter}}$ is a $L\times L$ block matrix divided into $N\times N$ matrices which we individually refer to with $\mathcal{L}^{\text{supra}}_{kl}$, $k,l=1,\dots,L$, and $\delta_{kl}$ is the Kronecker symbol.
\end{prop}
\begin{proof}
	Using $H=G$, the variation equation for~\eqref{eq:CoupledSysDiff} on the synchronous solution $\mathbf{s}(t)$ is given by
	\begin{align*}
		\dot{\bm{\xi}}=\left[\mathbb{I}_{NL}\otimes Df(\mathbf{s})-\mathcal{L}^{\text{supra}}\otimes H\right]\bm{\xi},
	\end{align*}
	By assumption all $L^{l}$, $l=1,\dots,L$, commute pairwise and $\mathcal{L}^{\text{inter}}$ is a $L\times L$ block matrix consisting of $N\times N$ identity matrices multiplied by scalar. Hence, Proposition~\ref{prop:DeterminantDecompDiagAble} can be applied to $\left(\mathcal{L}^{\text{supra}}-\mu\mathbb{I}_{L\cdot N}\right)$ in order to diagonalize $\mathcal{L}^{\text{supra}}$.
\end{proof}
With this Proposition, we have a very powerful tool in order to investigate not only the influence of the multiplex structure network on the stability of the synchronous state but also the impact of different layer topologies which are not necessarily linearly dependent. As an example, we consider a duplex systems with $[L^{1},L^{2}]=0$ and
\begin{align*}
	\mathcal{L}^{\text{supra}}=\begin{pmatrix}
		\sigma L^{1}+\rho m^{12}\mathbb{I} & -\rho m^{12}\mathbb{I}\\
		-\rho m^{21}\mathbb{I} & \sigma L^{2}+\rho m^{21}\mathbb{I}
	\end{pmatrix}.
\end{align*}
Knowing the eigenvalues for $L^{1}$ and $L^{2}$, the master parameter $\mu$ is determined using Eq.~\eqref{eq:DuplexDecomp}. The matrices $L^{1}$ and $L^{2}$ are Laplacian matrices which have at least one zero eigenvalue, corresponding to the so-called Goldstone mode $\hat{1}=(1,\dots N-\text{times}\cdots,1)^T$. As a result there are two parameters $\mu=0$ and $\mu=\rho(m^{12}+m^{21})$. The first value corresponds to the Goldstone mode $\begin{pmatrix}
\hat{1},
\hat{1}
\end{pmatrix}^T$. The second parameter is exclusively induced by the duplex structure and completely independent from the individual layer topologies. Further, we find that in case of $m^{12}=-m^{21}$ another zero parameter $\mu$ exists which corresponds to the eigenmode $\begin{pmatrix}
\hat{1},
-\hat{1}
\end{pmatrix}^T$. The other eigenvalues of the supra-Laplacian matrix $\mathcal{L}^{\text{supra}}$ are given by the non-linear mappings
\begin{multline}\label{eq:masterParamDuplex}
	\mu(\lambda^{1},\lambda^{2}) = \frac{\sigma(\lambda^{1}+\lambda^{2})+\rho(m^{12}+m^{21})}{2}\\
	\pm\frac12\sqrt{
		\begin{aligned}
			\left(\sigma(\lambda^{1}-\lambda^{2})+\rho(m^{12}-m^{21})\right)^2+ 4\rho^2m^{12}m^{21}
		\end{aligned}
	}
\end{multline}
for which we have formally solved equation~\eqref{eq:DuplexDecomp} with respect to $\mu$. Let us consider two special cases. 

First, we assume that there is no connection from the second to the first layer. We have a master-slave set-up which means that $m^{12}=0$. With this, the master parameter is $\mu(\lambda^{1},\lambda^{2}) =\sigma\lambda^{1}$ and $\mu(\lambda^{1},\lambda^{2}) = \sigma\lambda^{2}+\rho m^{21}$. Remarkably, in this set-up the stability of a synchronous state in a duplex network is reduced to the pure one-layer system. The stability in the duplex system is determined by the spectrum of the individual layer topologies where only in the second layer the spectrum is shifted due to the interaction.

The second case starts from the consideration in~\cite{Tang2019s}. In particular, we consider $\left[\mathcal{L}^{\text{intra}},\mathcal{L}^{\text{inter}}\right]=0$ which leads to a pairwise linear dependence of all individual layer topologies, and hence $\lambda_i^{l}=\lambda_i$ for all $l=1,\dots,L$ and $i=1,\dots,N$ with $\lambda_i^l\in\mathbb{C}$. Taking this into account, the equation for the master parameter yields $\mu(\lambda^{1},\lambda^{2}) =\sigma\lambda^{1}+\rho\left(m^{12}+m^{21}\right)$ and $\mu(\lambda^{1},\lambda^{2}) =\sigma\lambda^{1}$. In order to see that the master parameter agrees with the set for the parameter $\gamma$ in ~\eqref{eq:MSEComposite}, we determine the eigenvalues of $\mathcal{L}^{\text{intra}}$ and $\mathcal{L}^{\text{inter}}$ individually. Since $L^{(1)}$ and $L^{(2)}$ are linearly dependent, the set of eigenvalues for $\mathcal{L}^{\text{intra}}$ consists of the eigenvalues of $L^{1}$ with double multiplicity. Using Proposition~\ref{prop:EigenvaluePolynomsMultiplex}, we find that $\mathcal{L}^{\text{inter}}$ has eigenvalues $0$ and $(m^{12}+m^{21})$ each with multiplicity $N$. Since both Laplacian matrices commute, there exists a common set of eigenvectors. We find $\gamma=\rho\lambda^{(1)}$ for the eigenvector $(1,1)^T\otimes v$ and $\gamma=\rho\lambda^{(1)}+\sigma(m_{12}+m_{21})$ for the eigenvector $(m_{12},-m_{21})^T\otimes v$ where $L^{1}v=\lambda^{1}v$. Here, again everything boils down to a pure one-layer set-up with an additional shift.
\subsection{Analytic treatment of diffusive dynamics on multiplex networks}
In the previous section we considered the dynamics of linear systems as they are given by the variational equation~\eqref{eq:MSEComposite}. In the context of diffusive systems on complex networks, recently linear diffusive processes were considered in order to study the dynamics on social as well as transport networks~\cite{Barthelemy2011s,Gomez2013s}. Compared with equation~\eqref{eq:CoupledSysDiff} in~\cite{Gomez2013s}, the authors investigate a duplex system ($L=2$) with $f=0$, $\sigma=D_k$ ($k=1,2$), $H,G=\mathbf{I}$, and $\rho m^{kl}=D_x$. Additionally, they considere weighted networks for which instead of $a_{ij}^k\in\{0,1\}$ coupling weights $w_{ij}^{k}$ ($k=1,2$) were taken. Note that the former results on multiplex matrices still hold true for weighted connections.

Let us assume that the super-Laplacian matrix is given by
\begin{align*}
	\mathcal{L}^{\text{supra}}=\begin{pmatrix}
		D_1 L^{1}+D_x\mathbb{I} & -D_x\mathbb{I}\\
		-D_x \mathbb{I} & D_2 L^{2}+D_x\mathbb{I}
	\end{pmatrix}.
\end{align*}
with $[L^{1},L^{2}]=0$. Due to the structure of the supra-Laplacian we are allowed to apply~Proposition~\ref{prop:MSE_general}. In accordance with~\cite{Gomez2013s} and our former findings in Sec.~\ref{sec:MSA_MltPlx}, we have the two eigenvalues $\mu=0$ and $\mu=2D_x$ corresponding to the Goldstone mode $(\hat{1},\hat{1})^T$ and the vector $(\hat{1},-\hat{1})^T$, respectively. The other eigenvalues are given as solution to the equations~\eqref{eq:DuplexDecomp} and read
\begin{multline*}
	\mu_i = \frac{D_1\lambda_i^{1}+D_2\lambda_i^{2}+2D_x}{2}\\
	\pm\frac12\sqrt{
		\begin{aligned}
			\left(D_1\lambda_i^{1}-D_2\lambda_i^{2}\right)^2+ 4(D_x)^2
		\end{aligned}
	},
\end{multline*}
where $\lambda_i^1$ and $\lambda_i^2$ are the eigenvalues of the Laplacian matrices $L^{1}$ and $L^{2}$, respectively.
\subsection{Regions of synchronization in adaptive multiplex networks}
For an arbitrary duplex equilibrium of the form $\phi_i^{\mu}=\Omega t+a_i^{\mu}$  with $a_{k}^1=(0,\frac{2\pi}{N}k,\dots,(N-1)\frac{2\pi}{N}k)^{T}$ and $a_{k}^2=\mathbf{a}_{k}^1-\Delta a$ we start with the linearized system~(4) of the main text.
This can also be written in the block matrix form
\begin{align*}
	\begin{pmatrix}\dot{\delta\phi^1}\\
		\dot{\delta\phi^2} \\
		\dot{\delta\kappa^1} \\
		\dot{\delta\kappa^2}
	\end{pmatrix} & =\begin{pmatrix}A_1 & m_1\mathbb{I}_N & B_1 & 0\\
		m_2\mathbb{I}_N & A_2 & 0 & B_2\\
		C_1 & 0 & -\epsilon\mathbb{I}_{N^{2}} & 0\\
		0 & C_2 & 0 & -\epsilon\mathbb{I}_{N^{2}}
	\end{pmatrix}\begin{pmatrix}\delta\phi^1\\
		\delta\phi^2\\
		\delta\kappa^1\\
		\delta\kappa^2
	\end{pmatrix}
\end{align*}
with $\left(\delta\phi^\mu,\delta\kappa^\mu\right)^{T}=\left(\delta\phi_{1}^\mu,\dots,\delta\phi_{N}^\mu,\delta\kappa_{11}^\mu,\dots,\delta\kappa_{1N}^\mu,\right.$ $\left.\delta\kappa_{21}^\mu,\dots,\delta\kappa_{NN}^\mu\right)^T$, the matrices $A^{\mu}$, $B^{\mu}$, and $C^{\mu}$ follow from system~(4) of the main text, and $m_1, m_2 \in \mathbb{R}$. With the help of Schur's decomposition the characteristic equation for the linearized system takes the form
\begin{widetext}
	\begin{align}
		(\lambda+\epsilon)^{2(N^2-N)} &= 0\nonumber\\
		\det\begin{pmatrix}
			\left(\lambda\mathbb{I}_N-A_1\right)(\lambda+\epsilon)-B_1C_1 & -(\lambda+\epsilon)m_1\mathbb{I}_N \\
			-(\lambda+\epsilon)m_2\mathbb{I}_N & \left(\lambda\mathbb{I}_N-A_2\right)(\lambda+\epsilon)-B_2C_2 \\
		\end{pmatrix} &= 0. \label{eq:AKS_DublexSchur}
	\end{align}
\end{widetext}
The second equation has the block matrix form which is required from Proposition~\ref{prop:EigenvaluePolynomsMultiplex}. All blocks can be diagonalized and commute since they all possess a cyclic structure; compare Lemma~4.1 of~\cite{BER19s}. Thus, we are allowed to apply Proposition~\ref{prop:EigenvaluePolynomsMultiplex} which we use in order to diagonalize the matrix in Eq.~\eqref{eq:AKS_DublexSchur}. For the diagonalized matrix we find the following equations for the diagonal elements $\mu_i$
\begin{widetext}
	\begin{align}
		(\lambda+\epsilon)^2m_1 m_2 - (p_{i}^1(\lambda;\alpha^{11},\beta^1,\alpha^{12},\sigma^{12})-\mu_i)(p_{i}^2(\lambda;\alpha^{22},\beta^2,\alpha^{21},\sigma^{21})-\mu_i) = 0 \label{eq:EigenvalueMu}
	\end{align}
\end{widetext}
where $i=1,\dots,N$, $p_{i}^\mu(\lambda;\alpha^{\mu\mu},\beta^\mu)$ is a second order polynomial in $\lambda$ which depends continuously on $\alpha$ and $\beta$ as well as functionally on the type of the one-cluster state. For every $i\in\{1,\dots,N\}$, these equations will give us two eigenvalues $\mu_{i,1}$ and $\mu_{i,2}$ for the matrix in Eq.~\eqref{eq:AKS_DublexSchur} depending on $\lambda$ and the system parameters. Thus, we can write Eq.~\eqref{eq:EigenvalueMu} as
\begin{align*}
	\left(\mu_i-\mu_{i,1}(\lambda;\bm{\alpha},\bm{\beta},\bm{\sigma})\right)\left(\mu_i-\mu_{i,2}(\lambda;\bm{\alpha},\bm{\beta},\bm{\sigma})\right)=0
\end{align*}
where $\bm{\alpha},\bm{\beta},\bm{\sigma}$ represent all system parameter chosen for (\ref{eq:PhiDGL_general})\textendash (\ref{eq:KappaDGL_general}). In order to find the eigenvalue $\lambda$ of the linearized system~\eqref{eq:AKS_DublexSchur} one of the eigenvalues $\mu$ has to vanish. This means that we have to find $\lambda$ such that Eq.~\eqref{eq:EigenvalueMu} equals
\begin{align*}
	\mu_i\left(\mu_i-\mu_{i,2}(\lambda;\bm{\alpha},\bm{\beta},\bm{\sigma})\right)=0
\end{align*}
which is equivalent to finding $\lambda$ such that the following quartic equation is solved
\begin{widetext}
	\begin{align}
		p_{i}^1(\lambda;\alpha^{11},\beta^1,\alpha^{12},\sigma^{12})p_{i}^2(\lambda;\alpha^{22},\beta^2,\alpha^{21},\sigma^{21})-(\lambda+\epsilon)^2 m_1 m_2 = 0.
	\end{align}
\end{widetext}
Note that here the diagonal elements of $A_1$ are slightly different from those in Prop.~4.2 of \cite{BER19s} but they do not affect the result, i.e., the diagonal element equals $\rho_{i}(\alpha^{11},\beta^1)-m_1(\alpha_{12})$. The same holds true for $A_2$. Thus, with the two possible eigenvalues $\rho_{i,1,2}(\alpha^{\mu\mu},\beta^\mu)$ for the monoplex system from Corr.~4.3 of \cite{BER19s} one finds the following quartic equation which give the Lyapunov exponents for the lifted duplex one-cluster
\begin{widetext}
	\begin{multline}\label{eq:DuplexAKSLya}
		\left[\left(\lambda-\rho_{i,1}(\alpha^{11},\beta^1)\right)\cdot\left(\lambda-\rho_{i,2}(\alpha^{11},\beta^1)\right)+m_1(\lambda+\epsilon)\right]\times\\
		\left[\left(\lambda-\rho_{i,1}(\alpha^{22},\beta^2)\right)\left(\lambda-\rho_{i,2}(\alpha^{22},\beta^2)\right)+m_2(\lambda+\epsilon)\right]-(\lambda+\epsilon)^2 m_1 m_2 = 0.
	\end{multline}
\end{widetext}

In case of an antipodal monoplex one-cluster given by Eq.~(2) of the main text with $a_i\in \{0,\pi\}$, then the set of Lyapunov exponents $\mathcal{S}$ for Eq.~(4)
in the main text is given by

In the case of a duplex antipodal one-cluster state given by Eq.~(2) of the main text with $a^1_i\in \{0,\pi\}$ and $a_i^2=a_i^1-\Delta a$, Eq.\,(3) possesses the following set of Lyapunov exponents
\begin{align*}
	\mathcal{S}_{\text{Duplex}}=\{-\epsilon,\left(\lambda_{i,1},\lambda_{i,2},\lambda_{i,3},\lambda_{i,4}\right)_{i=1,\dots,N}\}
\end{align*}
where $\lambda_{i,1,\dots,4}$ solve the following $N$ quartic equations
\begin{eqnarray}
	\label{eq:DuplexEigenValues}
	(\lambda&+\epsilon)^2 m_1 m_2-\left[\left(\lambda-{\rho}_{i,1}^{1}\right)\cdot\left(\lambda-{\rho}_{i,2}^{1}\right)+m_1(\lambda+\epsilon)\right]\times\nonumber\\
	&\left[\left(\lambda-{\rho}_{i,1}^{2}\right)\left(\lambda-{\rho}_{i,2}^{2}\right)+m_2(\lambda+\epsilon)\right] = 0,
\end{eqnarray}
with $m_1\equiv\sigma^{12}\cos(\Delta a+\alpha^{12})$, $m_2\equiv\sigma^{21}\cos(\Delta a-\alpha^{21})$ and the eigenvalues ${\rho}_{i,1,2}^{\mu}\equiv{\rho}_{i,1,2}(\alpha^{\mu\mu},\beta^{\mu})$ for the monoplex system

\begin{multline}
	\label{eq:LyapAntipodal}
	\mathcal{S}_{\text{Monoplex}}=\left\{ \left(0\right)_{\text{1}},\left(-\epsilon\right)_{(N-1)N+1},\right.\\
	\left.\left(\rho_{1}\right)_{N-1},\left(\rho_{2}\right)_{N-1}\right\} 
\end{multline}
where $\rho_{1}$ and $\rho_{2}$ solve $\rho^{2}+\left(\epsilon-\cos(\alpha^{11})\sin(\beta^{1})\right)\rho-\epsilon\sin(\alpha^{11}+\beta^{1})=0$. Here, the multiplicities for each eigenvalue are given as lower case labels. The proof and the results for other clusters can be found in~\cite{BER19s,BER19as}.

Using the Lyapunov exponents of the duplex antipodal clusters, the stability for duplex antipodal states can be found. The results are presented in Fig.~4 (in the main text).

\end{document}